\documentclass[draftcls,11pt,onecolumn,romanappendices]{IEEEtran}
\let\proof\relax   

\usepackage{comment}
\usepackage{amsthm,xpatch}
\usepackage{amsmath,amsfonts}
\usepackage{cite}
\usepackage{amssymb}
\usepackage{dsfont}
\usepackage{graphicx, subfigure}
\usepackage{color}
\usepackage{breqn}
\usepackage{cancel}
\usepackage{mathtools}
\usepackage{bbm}
\usepackage{latexsym}
\usepackage[ruled, linesnumbered]{algorithm2e}
\usepackage{accents}
\usepackage{tikz}
\usepackage[mathscr]{eucal}
\usepackage{subfig}
\usepackage{booktabs}
\usetikzlibrary{calc}
\usetikzlibrary{math}
\usetikzlibrary{arrows.meta}
\usepackage{csquotes}
\usepackage{varioref}
\newtheorem{lemma}{Lemma}
\newtheorem{theorem}{Theorem}

\newcommand*{\transpose}{%
  {\mathpalette\@transpose{}}%
}

\IEEEoverridecommandlockouts

\begin{document}

\newcommand{\SB}[3]{
\sum_{#2 \in #1}\biggl|\overline{X}_{#2}\biggr| #3
\biggl|\bigcap_{#2 \notin #1}\overline{X}_{#2}\biggr|
}

\newcommand{\Mod}[1]{\ (\textup{mod}\ #1)}

\newcommand{\overbar}[1]{\mkern 0mu\overline{\mkern-0mu#1\mkern-8.5mu}\mkern 6mu}

\makeatletter
\newcommand*\nss[3]{%
  \begingroup
  \setbox0\hbox{$\m@th\scriptstyle\cramped{#2}$}%
  \setbox2\hbox{$\m@th\scriptstyle#3$}%
  \dimen@=\fontdimen8\textfont3
  \multiply\dimen@ by 4             
  \advance \dimen@ by \ht0
  \advance \dimen@ by -\fontdimen17\textfont2
  \@tempdima=\fontdimen5\textfont2  
  \multiply\@tempdima by 4
  \divide  \@tempdima by 5          
  \ifdim\dimen@<\@tempdima
    \ht0=0pt                        
    \@tempdima=\fontdimen5\textfont2
    \divide\@tempdima by 4          
    \advance \dimen@ by -\@tempdima 
    \ifdim\dimen@>0pt
      \@tempdima=\dp2
      \advance\@tempdima by \dimen@
      \dp2=\@tempdima
    \fi
  \fi
  #1_{\box0}^{\box2}%
  \endgroup
  }
\makeatother

\makeatletter
\renewenvironment{proof}[1][\proofname]{\par
  \pushQED{\qed}%
  \normalfont \topsep6\p@\@plus6\p@\relax
  \trivlist
  \item[\hskip\labelsep
        \itshape
    #1\@addpunct{:}]\ignorespaces
}{%
  \popQED\endtrivlist\@endpefalse
}
\makeatother

\makeatletter
\newsavebox\myboxA
\newsavebox\myboxB
\newlength\mylenA

\newcommand*\xoverline[2][0.75]{%
    \sbox{\myboxA}{$\m@th#2$}%
    \setbox\myboxB\null
    \ht\myboxB=\ht\myboxA%
    \dp\myboxB=\dp\myboxA%
    \wd\myboxB=#1\wd\myboxA
    \sbox\myboxB{$\m@th\overline{\copy\myboxB}$}
    \setlength\mylenA{\the\wd\myboxA}
    \addtolength\mylenA{-\the\wd\myboxB}%
    \ifdim\wd\myboxB<\wd\myboxA%
       \rlap{\hskip 0.5\mylenA\usebox\myboxB}{\usebox\myboxA}%
    \else
        \hskip -0.5\mylenA\rlap{\usebox\myboxA}{\hskip 0.5\mylenA\usebox\myboxB}%
    \fi}
\makeatother

\xpatchcmd{\proof}{\hskip\labelsep}{\hskip3.75\labelsep}{}{}

\pagestyle{plain}

\title{Increasing the Raw Key Rate in Energy-Time\\[-0.5ex] Entanglement Based Quantum Key Distribution}

\author{Esmaeil Karimi, Emina Soljanin, and Philip Whiting\thanks{E. Karimi is with the Department of Electrical and Computer Engineering, Texas A\&M University,
College Station, TX 77843 USA
(E-mail: esmaeil.karimi@tamu.edu).}
\thanks{E. Soljanin is with the Department of Electrical and Computer Engineering, Rutgers University, Piscataway, NJ 08854 USA
(E-mail: emina.soljanin@rutgers.edu).}
\thanks{P. Whiting is with the Department of Engineering, Macquarie University, North Ryde, NSW 2109 Australia (E-mail: philip.whiting@mq.edu.au).
}}


\maketitle 

\thispagestyle{plain}

\begin{abstract}
A Quantum Key Distribution (QKD) protocol describes how two remote parties can establish a secret key by communicating over a quantum and a public classical channel that both can be accessed by an eavesdropper. QKD protocols using energy-time entangled photon pairs are of growing practical interest because of their potential to provide a higher secure key rate over long distances by carrying multiple bits per entangled photon pair.
We consider a system where information can be extracted by measuring random times of a sequence of entangled photon arrivals. Our goal is to maximize the utility of each such pair.
We propose a discrete time model for the photon arrival process, and establish a theoretical bound on the number of raw bits that can be generated under this model. We first analyse a well known simple binning encoding scheme, and show that it generates significantly lower information rate than what is theoretically possible. We then propose three adaptive schemes that increase the number of raw bits generated per photon, and compute and compare the information rates they offer. Moreover, the effect of public channel communication on the secret key rates of the proposed schemes is investigated.
\end{abstract}

\section{introduction}
A Quantum Key Distribution (QKD) protocol describes how two parties, commonly referred to as Alice and Bob, can establish a secret key by communicating over a quantum and a public classical channel that both can be accessed by an eavesdropper Eve. For the widespread adoption of QKD, it is mandatory to provide high key rates over long distances (see a related survey\cite{Diamanti_2016}). What has appeared as a bottleneck in practice is the inability to maximize the utility of information-bearing quantum states that are communicated over the quantum channel\cite{article,Wehnereaam9288,PhysRL.Khan.2007}. 
QKD based on energy-time entangled photons has emerged as a promising technique primarily because each entangled photon pair can carry multiple raw key bits, and thus potentially provide a higher secure key rate over long distances \cite{lee2016highrate,Sarihan:19}. Moreover, it has been shown that higher dimensional quantum states are more sensitive to eavesdropping and are also more robust to certain types of
noise\cite{PhysRevLett.2002,PhysRevA.2005,PhysRevA.2006,PhysRevA.2010}.

Timing information extraction in energy-time entanglement based QKD schemes is commonly achieved through a method known as time-bin encoding \cite{Brougham_2013,Zhong2015PhotonefficientQK}. The time-bin encoding method is essentially a Pulse-Position Modulation (PPM) scheme, which is a common technique that converts the binary time-pulse sequences into large-alphabet sequences of fixed alphabet size. Alice and Bob timestamp their photon arrivals, and then map the timestamps to bit strings. Under ideal conditions, Alice and Bob are supposed to receive identical sequences.
The bit strings obtained in this case constitute the raw key. The objective of this paper is to maximize the length of the raw key.

Due to errors such as timing jitter, transmission loss and low detection efficiency, there are disparities between the received sequences in practical implementations\cite{Wornell.2013.Layered,10.5555/1972505,yang2020efficient}. In order to systematically increase the correlation between their key strings, while reducing Eve’s acquired information, Alice and Bob perform information reconciliation followed by privacy amplification, which reduces the key length\cite{Wornell.2013.Layered,yang2020efficient,10.5555/1972505,mao2019high}.
Note that achieving long raw keys does not necessarily imply long secret keys. A modulation scheme with a higher raw key might be more susceptible to noise and eavesdropping, and thus result in a relatively short secret key. Such considerations are beyond the scope of the current paper as we here are concerned only with the raw key rate.

A simple PPM scheme was proposed in \cite{6181805}. Although the simple PPM scheme eliminates the effect of photon transmission losses, it is not efficient for preserving useful information. In \cite{Wornell.2013.AdaPPM}, a generalized version of the simple PPM scheme, called adaptive PPM, was proposed which utilizes a good portion of the information discarded by the simple PPM scheme.  

In this work, our goal is to show that carefully modeled modulations can offer substantial raw key rate improvements, and also to pave the way for further exploration of high rate, low latency quantum-secure networks. We propose a new photon arrival model, a discrete time model for the photon arrival process with geometric distribution replacing the Poisson, and establish a theoretical bound on the number of secret bits that can be generated under this model (see Sec.~\ref{sec:SM}). Inspired by \cite{6181805,Wornell.2013.AdaPPM}, we first propose a simple binning scheme and show that this scheme generates significantly lower information rate than what is theoretically possible. We then propose three adaptive schemes that increase the number of raw bits generated per photon, and compute and compare the information rates they offer. Unlike the schemes in \cite{6181805,Wornell.2013.AdaPPM}, we not only use the single occupied bins but also utilize the single empty bins to generate secret bits (see Sec.~\ref{sec:PS}). Furthermore, we investigate the effect of public channel communication on the secret key rates of the proposed schemes (see Sec.~\ref{sec:PC}). 

\begin{figure*}[t]
    \centering
    \includegraphics[width=0.6\textwidth]{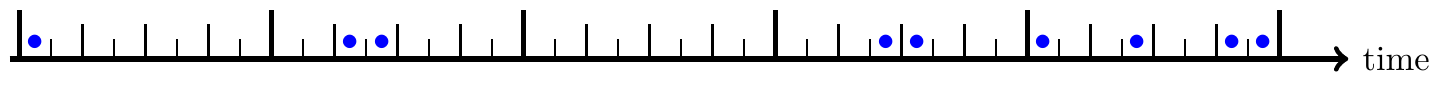}
    \caption{Five frames consisting of $8$ time units with bins consisting of $2$ time units.}
    \label{fig:PhotonLanding}
\end{figure*}

 \section{System Model}\label{sec:SM}
 Throughout the paper the base of $\log$ is $2$, unless explicitly noted otherwise. Consider a scenario wherein two parties, referred to as Alice and Bob, desire to generate a secret key using a quantum and a public channel. There is a third party, Eve, who has access to both channels.  A source (possibly co-located with Alice) emits entangled photon pairs to Alice and Bob, with one photon being sent to Alice, and the other to Bob. We consider a system where information can be extracted by measuring random times of a sequence of photon arrivals. We assume that time is measured in units such that at most one photon can arrive in a single time unit (See Fig.~\ref{fig:PhotonLanding}). The length of a time unit equals the minimum time that a photon detector needs to successfully detect a single photon. We assume that a photon arrives in each time unit with probability $p$ independently of other arrivals. The value of $p$ depends on the number of photons generated per second by the source. A similar model was adopted in \cite{Wornell.2013.AdaPPM,PhysRL.Khan.2007}. Photons are not fully utilized unless the arrival time of each received photon can be used to contribute information. 
\begin{theorem}\label{thrm:OptRate}
 The maximum number of bits per time unit that can be extracted from the timestamps of photon arrival times equals the binary entropy with parameter $p$. 
 \begin{equation}\label{eq:MaxRate}
   h(p)=-p\log p - (1-p)\log (1-p)   
 \end{equation}
\end{theorem}

 All the proofs can be found in the Appendix. Observe that this result implies that (under the assumed model) the photon timing information gives us as much information as would the binary sequence indicating the photon arrival times.

\section{Proposed Schemes \label{sec:PS}}
Considering the ideal case wherein all the incoming photons are transmitted and detected successfully,  Alice and Bob receive their shares of the entangled pairs at random but identical time units.
Alice and Bob timestamp their photon arrivals, and map these timestamps to bit strings, which they subsequently process to generate their common key. In this section, we ignore the effect of communication over the public channel on the raw key rate of a scheme. 
\subsection{Simple Binning}
In simple binning , time is partitioned in frames consisting of $n$ time units. We take $n$ to be a power of two. Fig.~\ref{fig:PhotonLanding} shows an example where $n=8$. Each frame is divided into $n/k$ bins, each consisting of $k\leq n$ time units, and we are free to choose $k$. Note that $k$ also needs to be a power of two in order for $n$ to be divisible by $k$. Bins are labeled by $\log (n/k)$ bit strings. A bin is called occupied if there is at least one photon present in the bin. Alice and Bob are able to generate a common random sequence based on the position of a single occupied bin or a single empty bin in the frames.

Information is extracted from a sequence of frames as follows: All frames are discarded except those containing either a single occupied bin or a single empty bin. Each frame with a single occupied bin contributes $\log (n/k)$ key bits identifying the single occupied bin since all positions of the occupied bin are equally likely. Similarly, each frame with a single empty bin contributes $\log (n/k)$ key bits. When there is one occupied bin and one empty bin, Alice and Bob consider the bit string label of the occupied bin as their common random sequence.  
Note that communication over the public channel is not needed here.

In the example of Fig.~\ref{fig:PhotonLanding}, if the bin size is chosen to be $2$, then the first 2 frames would contribute $2$ bits each since there is only one occupied bin among the four bins in each frame. The third and the fourth frames would be discarded since one is empty and the other one consists of two occupied and two empty bins. The fifth frame also contributes $2$ bits of information since it contains only one empty bin among its four bins. If, on the other hand,  the bin size is chosen to be $1$, then all but the first frame would be discarded, and we would be left with $3$ bits of information.

 The probability that a bin consisting of $k$ time units is occupied is given by ${\pi_k\triangleq 1-(1-p)^k}$. Let the probability that a bin consisting of $k$ time units is empty be given by ${\Bar{\pi}_k\triangleq (1-p)^k}$. We define the 
{\it raw key rate of a scheme} as the expected number of raw key bits per time unit.
 \begin{theorem}\label{theo:STB}
Let $n$ be the number of time units in a frame, and let each frame be divided into $n/k$ bins, each consisting of $k\leq n$ time units. The raw key rate of the simple binning scheme is given by
\begin{equation}\label{eq:SimpleRate}
    r_{SB}=
    \begin{cases}
0, &  k=n,\\
\frac{1}{k}\pi_k\Bar{\pi}_k , &  k=\frac{n}{2},\\
\frac{1}{k}\Big [ \pi_k\Bar{\pi}_k^{\frac{n}{k}-1}+\pi_k^{\frac{n}{k}-1}\Bar{\pi}_k\Big ]\log \frac{n}{k}, & \text{otherwise}.
\end{cases}
\end{equation}
\end{theorem}

We define {\it the photon utilization of a scheme} as the ratio between its raw key rate and the rate of the ideal scheme given by~\eqref{eq:MaxRate}.
Fig.~\ref{fig:SimpleRate} depicts the performance of the simple binning scheme. 
Two crucial parameters in simple binning encoding are the bin width and the frame size, which have to be carefully selected in order to maximize the photon utilization. The choice of these parameters also affects certain type of errors. It is therefore essential to understand the limitations that the system and physics impose on these parameters. Under no constraints, smaller bins and larger frames would maximize the photon utilization. However, physical constraints on energy-time entangled photons prevent the bin widths from becoming infinitely small. The minimum bin width is limited to the length of a time unit and the maximum frame size is limited by the coherence time of the entangled photon pair, which is determined by the spontaneous parametric down-conversion bandwidth \cite{PhysRL.Khan.2007,Zhong2015PhotonefficientQK}. Observe that, under restrictive conditions, the highest photon utilization achievable by the simple binning scheme will be limited, e.g., it is about $0.5$ for the frames of $n=64$ time units. This low efficiency is due to discarding a large fraction of frames. We next propose three more efficient schemes which use all or at least a large fraction of the frames. 

\subsection{Adaptive Binning}
The idea here is to not fix the size of the bins in advance, but instead adapt it to the photons observations for each frame. The size of the bins in a given frame is chosen by Alice and Bob deterministically based on the number and the locations of the photons observed in the frame as follows. Each bin is constructed using a collection of $k$ consecutive time units. The bin construction starts from the first time unit and ends at the last time unit in the frame. 
Alice and Bob choose the minimum $k$ that satisfies the following conditions: 1) the bins in a frame form a partition for the set of time units in the frame, and 2) either only one bin is occupied by photons among all the bins, or only one bin is empty among all the bins. We refer to these two conditions as the binning conditions. The rest follows the same steps as in the simple binning scheme. 

In this scheme, communication over the public channel is not required because the bin construction is done deterministically. In the example of Fig.~\ref{fig:PhotonLanding}, for the first frame, the minimum bin size that satisfies the binning conditions is $1$. Hence, the first frame contributes $3$ bits of information. The proper bin size for the second frame is $2$, and it contributes $2$ bits of information. The third frame is discarded. Let the time units in the fourth frame be labeled $1,\dots,8$. If we consider $k=2$, the bins will be $\{1,2\}$, $\{3,4\}$, $\{5,6\}$, and $\{7,8\}$. It is easy to see that the second and the third bins are occupied and the first and the fourth bins are empty. Thus, $k=2$ does not satisfy the binning conditions. If we let the bin size be $k=4$, we will be left with two occupied bins, and thus $k=4$ also does not satisfy the binning conditions. Hence, the minimum bin size for the fourth frame that satisfies the binning conditions is $k=8$. That is, the fourth frame consists of only one occupied bin. Thus, using this scheme, no information can be extracted from the fourth frame. The minimum bin size for the fifth frame that satisfies the binning conditions is $2$. There would be only one empty bin (third bin) among all four bins. Thus, the fifth frame also contributes $2$ bits of information.
\begin{theorem}\label{theo:AB}
Let $n$ be the number of time units and $\ell$ the number of photons in a frame. The raw key rate of the adaptive binning scheme is given by
\begin{equation}\label{eq:abrate}
r_{AB}=\sum_{\ell=1}^{n/2}\sum_{i={\lceil \log \ell \rceil}}^{\log (n/2)} \frac{1}{2^i}\binom{2^i}{\ell}p^{\ell}(1-p)^{n-\ell}+
\sum_{i=0}^{\log(n/4)}\frac{1}{2^i}\pi_{2^i}^{\frac{n}{2^i}-1}\Bar{\pi}_{2^i}(\log n-i).
\end{equation}
\end{theorem}

\subsection{Adaptive Aggregated Binning}
In this scheme, the size of the bins in individual frames depends only on the number of photons observed in the frame. When a frame is occupied with $\ell\leq n/2$ photons, Alice partitions the set of time units in the frame into ${m={n}/{2^{\lceil \log \ell \rceil}}}$ bins of size $k=2^{\lceil \log \ell \rceil}$, denoted by $B_1,B_2,\dots,B_m$. Then, Alice chooses a bin randomly, say $B_i$, and  assigns all the $\ell$ time units carrying a photon to this bin. Also, from the remaining time units, $k-\ell$ randomly chosen time units will be assigned to $B_i$. After this step, from the remaining time units, $k$ randomly picked time units will be assigned to each bin $B_j$ for $j\in\{1,2,\cdots,m\}\setminus i$. Note that there exists only one occupied bin and the position of this bin is uniformly distributed.

Otherwise, when $\ell > n/2$ photons have been observed in a frame, Alice partitions the set of the time units in the frame into ${m={n}/{2^{\lfloor \log (n-\ell) \rfloor}}}$ bins of size $k=2^{\lfloor \log (n-\ell) \rfloor}$, denoted by $B_1,B_2,\dots,B_m$. Then, Alice chooses a bin randomly, say $B_i$, and  assigns $k$ randomly picked empty time units to this bin. From the remaining time units, Alice assigns $k$ randomly chosen to each bin $B_j$ for $j\in\{1,2,\cdots,m\}\setminus i$. Note that there exists only one empty bin and the position of this bin is uniformly distributed. After forming the bins, Alice sends the binning information to Bob over the public channel.

In the example of Fig.~\ref{fig:PhotonLanding}, the first frame contributes $3$ bits of information. The second and the fourth frames contribute $2$ bits of information each since Alice is able to form $4$ bins of size $2$ where only one of the bins is occupied. The third frame would be discarded. The fifth frame contributes $1$ bit of information since the time units in the frame can be partitioned into $2$ bins of size $4$ while only one of the bins is occupied.

 \begin{theorem}\label{theo:AAB}
Let $n$ be the number of time units and $\ell$ the number of photons in a frame. The raw key rate of the adaptive aggregated binning scheme is given by 
\begin{equation}
r_{AAB}=\frac{1}{n}\Biggr [\sum_{\ell=1}^{n/2}\binom{n}{\ell}p^{\ell}(1-p)^{n-\ell}\Big(\log n - {\lceil \log \ell \rceil}\Big)+
\sum_{\ell=\frac{n}{2}+1}^{n-1}\binom{n}{\ell}p^{\ell}(1-p)^{n-\ell}\Big(\log n - {\lfloor \log (n-\ell) \rfloor}\Big)\Biggr ]. 
\label{eq:aabRate}   
\end{equation}
\end{theorem}

\subsection{Adaptive Framing} 
Unlike the other schemes, in this scheme, the bin size do not vary from frame to frame and for all the frames is $k=1$. Having observed $\ell\leq n/2$ photons in a frame, the set of time units in the frame will be partitioned into $\ell$ subframes by Alice. It should be noted that a subframe does not consist of adjacent time units necessarily. Let $F_1,F_2,\dots,F_{\ell}$ denote these subframes, and let $i_1,i_2,\dots,i_{\ell}$ be the indices of the time units carrying a photon. At the beginning, Alice assigns the time unit $i_j$ to the subframe $F_j$ for $j\in \{1,2,\cdots,\ell\}$. Then, starting from the first subframe, each subframe randomly picks an unassigned time unit. The previous step will be done repeatedly until all the time units have been assigned. In each subframe, there is exactly one bin occupied with a photon and its position is uniformly distributed. This procedure results in $r$ subframes of size $m+1$ and $\ell-r$ subframes of size $m$, where ${n=m\ell+r}$ and $0\leq r <\ell$. Hence, each frame occupied with $\ell\leq n/2$ photons contributes ${\rho =r\log (m+1)+(\ell-r)\log m}$ bits of information. The following lemma shows that this is the maximum information that can be extracted from a frame of size $n$ containing $\ell \leq n/2$ photons using the adaptive framing scheme.
\begin{lemma}\label{lem:MaxRate1}
Let $n$ be the size of a frame consisting of ${ \ell \leq n/2}$ photons. Alice constructs $\ell$ sets and assigns one each of the occupied time units to the respective sets. The remaining time units are assigned at random to the sets. Let $d_i\geq 1$ denote the number of elements in set $i$. The total information that can be extracted from the frame is therefore $I= \sum_{i=1}^{\ell}\log d_i$. It holds that \[I= \sum_{i=1}^{\ell}\log d_i \leq r\log (m+1)+(\ell-r)\log m,\]
where $n=m\ell+r$ and $0\leq r <\ell$.
\end{lemma}

On the other hand, when the number of photons observed in a frame is $\ell>n/2$, Alice partitions the set of time units in the frame into $n-\ell$ subframes. Let $F_1,F_2,\dots,F_{n-\ell}$ denote these subframes, and let $i_1,i_2,\dots,i_{n-\ell}$ be the indices of the empty time units. First, the time unit $i_j$ is assigned to the subframe $F_j$ for $j\in \{1,2,\cdots,n-\ell\}$ by Alice. Then, each subframe chooses an unassigned time unit at random starting from the first subframe. This step will be repeated until all time units have been assigned. In the end, there are $\Bar{r}$ subframes of size $\Bar{m}+1$ and $n-\ell-\Bar{r}$ subframes of size $\Bar{m}$, where ${n=\Bar{m}(n-\ell)+\Bar{r}}$ and $0\leq \Bar{r} <n-\ell$. There is exactly one empty time unit in each subframe, and its position is uniformly distributed. Thus, each frame occupied with ${\ell > n/2}$ photons contributes $\Bar{\rho}={\Bar{r}\log (\Bar{m}+1)+(n-\ell-\Bar{r})\log \Bar{m}}$ bits of information. Using the following lemma, we show that this is the maximum information that can be extracted from a frame of size $n$ containing $\ell > n/2$ photons using the adaptive framing scheme.
\begin{lemma}\label{lem:MaxRate2}
Let $n$ be the size of a frame consisting of ${ \ell > n/2}$ photons. Alice constructs $n-\ell$ sets and assigns one each of the empty time units to the respective sets. The remaining time units are assigned at random to the sets. Let $d_i\geq 1$ denote the number of elements in set $i$. The total information that can be extracted from the frame is therefore $I= \sum_{i=1}^{n-\ell}\log d_i$. It holds that \[I= \sum_{i=1}^{n-\ell}\log d_i \leq \Bar{r}\log (\Bar{m}+1)+(n-\ell-\Bar{r})\log \Bar{m},\]
where $n=\Bar{m}(n-\ell)+\Bar{r}$ and $0\leq \Bar{r} <n-\ell$.
\end{lemma}

The subframes information will be sent to Bob over the public channel by Alice. In the example of Fig.~\ref{fig:PhotonLanding}, the first frame contributes $3$ bits of information. The second and the fourth frames contribute $4$ bits of information each. For instance, consider the second frame. Let index the time units in the second frame using the numbers 1 to 8. The time units $3$ and $4$ are occupied with photons. Alice forms two subframes denoted by $F_1$ and $F_2$, and assigns the time units $3$ and $4$ to these two subframes, respectively. Then, from the remaining time units, Alice assigns $3$ time units to each subframe randomly as it was explained before, and sends the subframes information to Bob over the public channel. Thus, Alice and Bob have information about two subframes containing four time units while only one time units carries a photon in each subframe. These two subframes contribute $2$ bits of information each. The third frame is discarded, and the fifth frame contributes $4$ bits of information since it can be partitioned into $4$ subframes of size $2$ where there is one occupied time unit in each subframe. 
 \begin{theorem}\label{theo:AF}
Let $n$ and $\ell$ denote the number of time units and the number of photons in a frame, respectively. The raw key rate of the adaptive framing scheme is given by 
\begin{equation}\label{eq:afRate}
    r_{AF}=\frac{1}{n}\Biggr [\sum_{\ell=1}^{n/2}\binom{n}{\ell}p^{\ell}(1-p)^{n-\ell}\rho
    +\sum_{\ell=\frac{n}{2}+1}^{n-1}\binom{n}{\ell}p^{\ell}(1-p)^{n-\ell}\Bar{\rho}
    \Biggr ].
\end{equation}
 \end{theorem}

\vspace{0.1cm}
\section{Effect of Public Channel Communication}\label{sec:PC}
In this section, we investigate the effect of public channel communication on the raw key rate. 
For the simple binning and the adaptive binning, communication over the public channel is not required. However, in the adaptive aggregated binning and adaptive framing, after each time frame, Alice needs to form bins or subframes and send the information to Bob over the public channel. Thus, for these two schemes, we partition time into a number of windows, which we further split into two phases: sensing phase and communication phase. In the sensing phase, which consists of $n$ time units, Alice and Bob observe photon arrival times ,and in the communication phase, they talk over the public channel. Let $D$ and $T$ denote the communication time over the public channel and the length of a time unit, respectively. Hence, the length of a window is $nT+D$ and the number of raw secret bits that a scheme generates in a window is given by $n\times $(raw key rate of the scheme). We define the 
{\it effective raw key rate of a scheme} as the expected number of raw key bits per time unit considering the effect of public channel communication. The raw key rates and the effective raw key rates of the simple binning and the adaptive binning schemes are the same. The effective raw key rate of the adaptive aggregated binning and adaptive framing schemes are given as follows
\vspace{0.1cm}
\begin{equation*}
    \Tilde{r}_{AAB}=\frac{nT}{nT+D}r_{ABB},
\end{equation*}    

\begin{equation*}
    \Tilde{r}_{AF}=\frac{nT}{nT+D}r_{AF}.
\end{equation*}
 Note that the typical length of a time unit is about tens of picoseconds ($10^{-12}$ seconds)\cite{lee2016highrate,Zhong2015PhotonefficientQK}.
 
\begin{figure*}[t]
    \centering
    \includegraphics[width=1\textwidth]{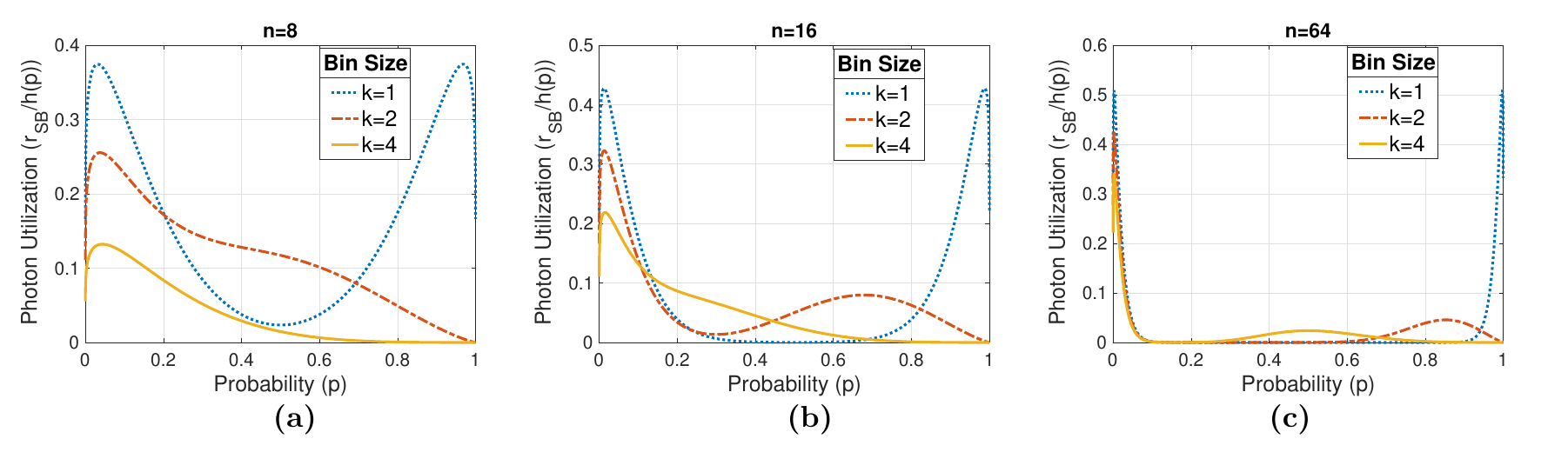}
    
    \vspace{-0.5cm}
    \caption{\small{Average photon utilization of the simple binning scheme vs.\ the time unit occupancy probabilities (cf.~\eqref{eq:MaxRate}~and~\eqref{eq:SimpleRate}), for: a) frames of $n=8$ time units and three different bin sizes $k\in\{1,2,4\}$, b) frames of $n=16$ time units and three different bin sizes $k\in\{1,2,4\}$, and c) frames of $n=64$ time units and three different bin sizes $k\in\{1,2,4\}$.}}
    \label{fig:SimpleRate}
\end{figure*}

\begin{figure*}[t]
    \centering
    \includegraphics[width=1\textwidth]{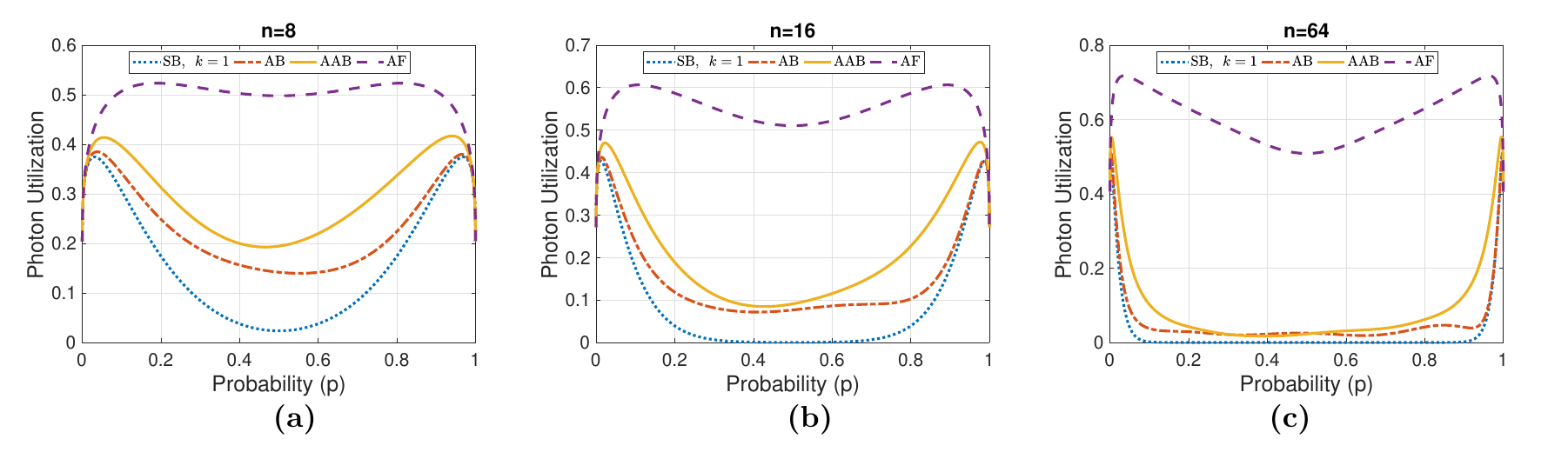}
    
    \vspace{-0.4cm}
\caption{\small{Average photon utilization of the simple binning (SB) for bin size $k=1$ (cf.~\eqref{eq:SimpleRate}), the adaptive binning (AB) (cf.~\eqref{eq:abrate}), the adaptive aggregated binning (AAB) (cf.~\eqref{eq:aabRate}), and the adaptive framing (AF) (cf.~\eqref{eq:afRate}) schemes vs.\ the time unit occupancy probability, for: a) frames of $n=8$ time units, b) frames of $n=16$ time units, and c) frames of $n=64$ time units.}}
\label{fig:Comp}
\end{figure*} 
 
\section{Comparison Results}\label{sec:CR}
In this section, we evaluate and compare the performance of the proposed schemes. Fig.~\ref{fig:SimpleRate} illustrates the performance of the simple binning scheme. It can be observed that for all three different frame sizes, the maximum photon utilization is achieved when the bin size is set to $1$. It can also be seen that increasing the frame size improves the highest achievable photon utilization for all three different bin sizes. Note that, for some range of the time unit occupancy probability, bin sizes $k=2$ and $k=4$ result in a higher photon utilization in comparison to bin size $k=1$. 

The photon utilization of the simple binning (SB) for bin size $k=1$, the adaptive binning (AB), the adaptive aggregated binning (AAB), and the adaptive framing (AF) schemes as a function of the time unit occupancy probability is depicted in Fig.~\ref{fig:Comp}. Observe that the AF outperforms the other three schemes for all range of the time unit occupancy probability. For all four schemes, the highest photon utilization is obtained when the time unit occupancy probability is either close to $0$ or close to $1$. Moreover, the performances of all the schemes are identical when time unit occupancy probability is very small or very large, since almost all the occupied frames carry $1$ photon or $n-1$ photons, respectively. Note that, although the AF and the AAB have a superior performance in comparison to the SB and the AB, they require public channel communication.

\bibliographystyle{IEEEtran}
\bibliography{Long-version}

\appendix[Proof of Lemmas and Theorems]

\begin{proof}[Proof of Theorem~\ref{thrm:OptRate}]
 The photon inter-arrival times are geometrically distributed. Thus, the maximum information that can be extracted from an observed photon is equal to the entropy of a geometric random variable with parameter $p$, which is given by $[-p\log p - (1-p)\log (1-p)]/p$. In the period of $n$ time units, the average number of observed photons is equal to $np$, and thus the average number of bits that can be extracted in the period of $n$ time units is $n[-p\log p - (1-p)\log (1-p)]$. Hence, the number of bits per time unit that can be obtained on average is given by ${h(p)=-p\log p - (1-p)\log (1-p)}$.
 \end{proof}
 
 \begin{proof}[Proof of Theorem~\ref{theo:STB}]
Let $A$ denote the event that there is only one occupied bin in a frame. The probability of event $A$ is given by $P(A)=\binom{n/k}{1}\pi_k\Bar{\pi}_k^{{n}/{k}-1}$. Also, let the event that only one empty bin exists in a frame be denoted by $B$. The probability of event $B$ is given by ${P(B)=\binom{n/k}{1}\pi_k^{{n}/{k}-1}\Bar{\pi}_k}$. The raw key rate of the simple binning scheme is given by ${r_{SB}=\frac{1}{n}P(A\cup B)\log (n/k)}$. For the case $k=n$, we have $\log (n/k)=0$ and consequently $r_{SB}=0$. The case $k={n}/{2}$ indicates that there are two bins of size $n/2$ in a frame. One can readily confirm that the events $A$ and $B$ are equivalent for this case. Thus, we have ${P(A)=P(B)=P(A\cap B)}$. Note that $\log (n/k)=1$ and $P(A\cup B)=P(A)=\frac{n}{k}\pi_k\Bar{\pi}_k$ when $k={n}/{2}$. Therefore, $r_{SB}=\frac{1}{k}\pi_k\Bar{\pi}_k$ for the case $k=n/2$. For the cases where $k\leq n/4$, we have $P(A\cap B)=0$ and consequently $P(A\cup B)=P(A)+P(B)=\frac{n}{k}\Big [ \pi_k\Bar{\pi}_k^{\frac{n}{k}-1}+\pi_k^{\frac{n}{k}-1}\Bar{\pi}_k\Big ]$. Thus, $r_{SB}=\frac{1}{k}\Big [ \pi_k\Bar{\pi}_k^{\frac{n}{k}-1}+\pi_k^{\frac{n}{k}-1}\Bar{\pi}_k\Big ]\log \frac{n}{k}$ when $k\leq n/4$.
\end{proof}
\begin{proof}[Proof of Theorem~\ref{theo:AB}]
Given that $\ell$ photons have been observed in a frame, the probability that bins of size $k$ satisfy the binning conditions such that there is only one occupied bin in the frame is given by $p_k(\ell)=\binom{n/k}{1}\binom{k}{\ell}p^{\ell}(1-p)^{n-\ell}$. When $k < \ell$, it is assumed that $p_k(\ell)=0$. Note that $k$ is not necessarily the minimum bin size that satisfy the binning conditions, and thus $p_k(\ell)$ includes all the cases that $k/2^i$, $i\in\{0,1,\cdots,\log k\}$, is the minimum bin size that satisfies the binning conditions. Hence, the probability that $k$ is the minimum bin size that satisfies the binning conditions such that there is only one occupied bin in the frame is given by ${p_k(\ell)-p_{k/2}(\ell)}$. The number of bits obtained by the cases wherein there is only one occupied bin in the frame is given by $\sum_{\ell=1}^{n/2}\sum_{i={\lceil \log \ell \rceil}}^{\log (n/2)}\Big( p_{2^i}(\ell)-p_{2^{i-1}}(\ell) \Big) (\log n -i)$. 
We can simplify $\sum_{i={\lceil \log \ell \rceil}}^{\log (n/2)}\Big( p_{2^i}(\ell)-p_{2^{i-1}}(\ell) \Big) (\log n -i)$ by expanding it as follows. Let $x\triangleq {\lceil \log \ell \rceil}$ and $y\triangleq \log n$. Note that $p_{2^{x-1}}(\ell)=0$ since $2^{x-1}<\ell$.

\begin{multline*}
 \sum_{i=x}^{y-1}\Big( p_{2^i}(\ell)-p_{2^{i-1}}(\ell) \Big) (y -i)= {\color{red}\cancel{p_{2^x}(\ell)(y-x)}}\\
 {\color{blue}\cancel{+p_{2^{x+1}}(\ell)(y-x-1)}}{\color{red}\cancel{-p_{2^x}(\ell)(y-x)}}+p_{2^x}(\ell)\\
 {\color{orange}\cancel{+p_{2^{x+2}}(\ell)(y-x-2)}}{\color{blue}\cancel{-p_{2^{x+1}}(\ell)(y-x-1)}}+p_{2^{x+1}}(\ell)\\
 \vdots\\
 +p_{2^{y-1}}(\ell)(y-y+1){\color{magenta}\cancel{-p_{2^{y-2}}(\ell)(y-y+2)}}+p_{2^{y-2}}(\ell)\\
 =p_{2^x}(\ell)+p_{2^{x+1}}(\ell)+\dots+p_{2^{y-2}}(\ell)+ p_{2^{y-1}}(\ell)=\sum_{i=x}^{y-1}p_{2^i}(\ell)\\
\end{multline*}

\vspace{-0.4cm}
The probability that $k$ is the minimum bin size that satisfies the binning conditions such that there is only one empty bin in the frame is given by $\binom{n/k}{1}\pi_k^{\frac{n}{k}-1}\Bar{\pi}_k$, where ${k< n/2}$. Note that $k=n/2$ has already been addressed as it is the same for the case that there is only one occupied bin in the frame.
The number of bits obtained by the cases wherein there is only one empty bin in the frame is given by $\sum_{i=0}^{\log(n/4)}\binom{n/2^i}{1}\pi_{2^i}^{\frac{n}{2^i}-1}\Bar{\pi}_{2^i}\log(\frac{n}{2^i})$. Thus, the raw key rate of the adaptive binning scheme is given by
\begin{equation*}
r_{AB}=\frac{1}{n}\Biggr [\sum_{\ell=1}^{n/2}\sum_{i={\lceil \log \ell \rceil}}^{\log (n/2)} p_{2^i}(\ell)+
\sum_{i=0}^{\log(n/4)}\binom{n/2^i}{1}\pi_{2^i}^{\frac{n}{2^i}-1}\Bar{\pi}_{2^i}\log(\frac{n}{2^i})\Biggr].
\end{equation*}
\end{proof}

\begin{proof}[Proof of Theorem~\ref{theo:AAB}]
In the adaptive aggregated binning scheme, when a frame contains ${\ell\leq n/2}$ photons, the time units in the frame are partitioned into $m={n}/{2^{\lceil \log \ell \rceil}}$ bins of size $2^{\lceil \log \ell \rceil}$ such that only one of the bins is occupied. Thus, each frame containing $\ell\leq n/2$ photons contributes ${\log m=\log n-{\lceil \log \ell \rceil}}$ bits of information. The probability that $\ell$ photons are observed in a frame is given by $\binom{n}{\ell}p^{\ell}(1-p)^{n-\ell}$. Using a similar argument, one can show that each frame consisting of $\ell>n/2$ photons contributes ${\log n-{\lfloor \log (n-\ell) \rfloor}}$ bits of information. Thus, it is easy to see that~\eqref{eq:aabRate} gives the raw key rate of the adaptive aggregated binning scheme.
\end{proof}

\begin{proof}[Proof of Lemma~\ref{lem:MaxRate1}]
If $r=0$, this inequality is an immediate consequence of Jensen’s inequality and the concavity of the log function. Hence, suppose $r>0$. There must be at least one $i$ such that $d_i\leq m$ as otherwise ${n<\ell (m+1)\leq \sum_{i=1}^{\ell}\log d_i}$ which contradicts
that the $d_i$'s sum to $n$. Similarly, there is an $i$ such that $d_i\geq m+1$. We will now show that if there is an $i$ such that $d_i<m$ or $d_i> m+1$, then $I$ can be strictly increased. First, suppose that there is an $i_j$ such that $d_{i_j}<m$, and an $i_h$ such that $d_{i_h}>m+1$. We may suppose these correspond to the largest and the smallest sets. Then, take an empty time unit from a set of size $d_{i_h}$ and place it in one of the sets of size $d_{i_j}$. Since the log function is strictly increasing and strictly concave, we gain ${(\log(d_{i_j}+1)-\log d_{i_j})-(\log d_{i_h}-\log (d_{i_h}-1))> 0}$. Clearly, such exchanges can continue until either all sets have at least $m$ members or no set has more than $m+1$ members. If all sets are of size $m$ or $m+1$, then the argument is complete. Now, suppose that there is a set with more than $m+1$ time units with the remaining sets having $m$. Then, the number of sets of size $m$ must be equal to $\ell-r+f$ with $f>0$ as the total number of time units is equal to $n$. Now, take an empty time unit from the set with largest size and place it in a set of size $m$, which gives an increase in information as before. Repeat this until $f$ becomes $0$ so that $I$ becomes $r\log (m+1)+(\ell-r)\log m$. A similar argument applies if there is a set $d_i<m$, and the remaining sets all have $m+1$ time units. 
\end{proof}

\begin{proof}[Proof of Lemma~\ref{lem:MaxRate2}]
The proof is similar to the proof of Lemma~\ref{lem:MaxRate1}, and thus omitted for the purpose of brevity.
\end{proof}

\begin{proof}[Proof of Theorem~\ref{theo:AF}]
 It has been already shown that, in the adaptive framing scheme, each frame occupied with ${\ell\leq n/2}$ photons contributes ${\rho =r\log (m+1)+(\ell-r)\log m}$ bits of information, where ${n=m\ell+r}$ and $0\leq r <\ell$. Also, it has been shown that each frame occupied with ${\ell > n/2}$ photons contributes $\Bar{\rho}={\Bar{r}\log (\Bar{m}+1)+(n-\ell-\Bar{r})\log \Bar{m}}$ bits of information, where ${n=\Bar{m}(n-\ell)+\Bar{r}}$ and ${0\leq \Bar{r} <n-\ell}$. The probability that $\ell$ photons are observed in a frame is given by $\binom{n}{\ell}p^{\ell}(1-p)^{n-\ell}$. Thus, it is easy to see that~\eqref{eq:afRate} gives the raw key rate of the adaptive framing scheme.
 \end{proof}
\end{document}